\documentclass[letterpaper,USenglish,cleveref]{lipics-v2019}

\usepackage{url}
\usepackage{wrapfig}
\usepackage{xcolor}

\newcommand{\conv}{\text{conv}}
\newcommand{\eps}{\varepsilon}

\title{On Mergable Coresets for Polytope Distance}

\author{Benwei Shi, Aditya Bhaskara, Wai Ming Tai, and Jeff M. Phillips}{School of Computing, University of Utah}{b.shi@utah.edu, bhaskara@cs.utah.edu, u1008421@utah.edu, jeffp@cs.utah.edu}{}{}

\authorrunning{B. Shi, A. Bhaskara, W. Tai and J.\,M. Phillips}  

\Copyright{Benwei Shi, Aditya Bhaskara, Waiming Tai and Jeff M. Phillips}



\category{}  

\nolinenumbers 

\hideLIPIcs  

\EventEditors{John Q. Open and Joan R. Access}
\EventNoEds{2}
\EventLongTitle{42nd Conference on Very Important Topics (CVIT 2016)}
\EventShortTitle{CVIT 2016}
\EventAcronym{CVIT}
\EventYear{2016}
\EventDate{December 24--27, 2016}
\EventLocation{Little Whinging, United Kingdom}
\EventLogo{}
\SeriesVolume{42}
\ArticleNo{23}

\begin{document}
	
\maketitle


\section{Introduction.}
The max-margin linear separator is a classic problem in machine learning~\cite{Burges1998}, defined as follows.  
Given a point set $P \subset \mathbb{R}^d$ with labels $\{-1, +1\}$ find a hyperplane $h$ that separates the labels, which maximizes the \emph{margin} $\gamma = \min_{p \in P} \|p - \pi_h(p)\|$, where $\pi_h(p)$ projects $p$ onto $h$.
This is equivalent to the two-polytope min-distance problem, and can be reduced to the one-polytope min-distance (polytope distance for short) problem~\cite{Gartner2009}. 
Further, a $(1-\eps)$-approximation of polytope distance can be used to obtain a $(1-\eps)$-approximation for max-margin separating hyperplane. 
The former can be solved by finding an $\eps$-coreset---the objective of this paper, defined formally below. 

In this paper we ask if these $\eps$-coresets can be merged~\cite{ACHPWY}.
That is, given two $\eps$-coresets $S_1$ and $S_2$, can they be combined into a single $\eps'$-coreset while not increasing the space (hopefully with $\eps' = \eps$).  
By creating coresets on batches of points in a streaming setting, if we can iteratively merge these coresets, this easily leads to streaming algorithms.  This framework also implies small space and communication complexity in other big data settings.

\section{Concepts and Definitions.}
\setlength\intextsep{0pt}
\begin{wrapfigure}{r}{0pt}
    \vspace{-2cm}
    \includegraphics{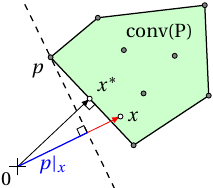}
    \caption{$P$ in gray, and $\conv(P)$ in green.  Point $x \in \conv(P)$ is an $(1-\eps)$-approximation: the red part has length $\eps \|x\|$. }
    \label{fig:eps-approx}
\end{wrapfigure}	
We follow the definition of $\eps$-coreset for polytope distance problem used in G\"artner and Jaggi's paper \cite{Gartner2009}.
Formally, we are given a point set $P \in \mathbb R^d$, we want approximate $x^* = \arg\min_{v \in \conv(P)} \|v\|$, the point in $\conv(P)$ closest to the origin. 
Define $p|_x := \frac{\langle p,x \rangle}{\|x\|}$ as the signed length of the \textit{projection} of $p$ onto the direction of the vector $x$. 
For any $\eps > 0$, $x \in \conv(P)$ is called an \textit{$\eps$-approximation}, iff $(1-\eps)\|x\| \le p|_x, \ \forall p \in P$; see \Cref{fig:eps-approx}.
This approximation is stronger than just requiring the distance $\|x\|$ to be close to the optimal value, $(1 - \eps)\|x\| \le \|x^*\|$. In particular, if $x$ is an $(1-\eps)$-approximation, it implies that
$ (1 - \eps)\|x\| \le \min_{p \in P} p|_x = \min_{v \in \conv(P)} v|_x \le \|x^*\| \le \|x\|$.
A subset $S \subseteq P$ is an \textit{$\eps$-coreset} of $P$ iff $\conv(S)$ contains an $(1-\eps)$-approximation to the distance of $\conv(P)$.

To bound the $\eps$-coreset size, we need some bound on the width of the data $P$.  G\"artner and Jaggi~\cite{Gartner2009} use the \emph{excentricity} of a point set $P$, defined $E = \frac{\mathsf{diam}(\conv(P))^2}{\|x^*\|^2}$.  
An $\eps$-coreset with size no more than $2 \lceil 2 \frac{E}{\eps} \rceil$ always exists~\cite{Gartner2009,Burges1998}, and can be found with a simple greedy (Frank-Wolfe) algorithm. 
In this paper we use the \emph{angular diameter} $\theta$ instead of the excentricity $E$; it is defined as the maximum angle between any two vectors (points) from $P$.
While incomparable to excentricity, this property allows us to provide upper and lower bounds on the mergeability of polytope distance coresets.
	
\section{Our results.}
We announce mainly negative results.  
First we show a constant-size $(1-\cos \theta)$-coreset for polytope distance is simple to find and maintain under merges (\Cref{thm:shortest}). However, increasing the size of the coreset cannot significantly improve the error bound (\Cref{thm:coreset,thm:0coreset}); we cannot maintain $\eps$-coresets with arbitrarily small $\eps > 0$ under merges. 

This hardness is not totally unexpected given the known hardness of streaming $(1-\eps)$-approximate minimum enclosing ball~\cite{Sharath+Agarwal}, which would also imply a streaming coreset for max-margin linear separators.  We show that even if we restrict $P$ to have angular diameter at most $\pi/2$ (the margin is relatively large), polytope distance $\eps$-coresets cannot be used to derive a $(1-\eps)$-approximate margin algorithm for max-margin linear separators.  

\subparagraph{Maintaining a simple coreset.}
We first show that the closest point is an $\eps$-coreset with $\eps = 1 - \cos\theta$.  This is trivial to maintain under merges.  

\begin{figure}
    \minipage{0.3\textwidth}
    \includegraphics{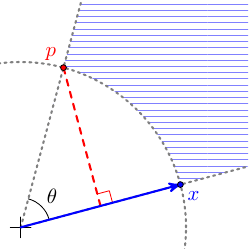}
    \caption{Using shortest point $x$ as coreset}
    \label{fig:shortest-point}
    \endminipage\hfill
    \minipage{0.3\textwidth}
    \includegraphics{"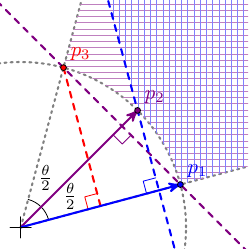"}
    \caption{merging two $(1-\cos\frac\theta 2)$-coresets}
    \label{fig:epsilon-coreset}
    \endminipage\hfill
    \minipage{0.3\textwidth}%
    \includegraphics{"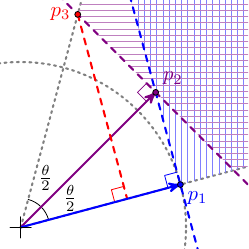"}
    \caption{merging two 0-coresets}
    \label{fig:0coreset}
    \endminipage
\end{figure}

\begin{theorem}
    \label{thm:shortest}
    Consider a point set $P$ with angular diameter $\theta \le \frac{\pi}{2}$. Let $x = \arg\min_{p \in P} \|p\|$, then $x$ is a $(1-\cos \theta)$-coreset of $P$.
\end{theorem}
\begin{proof}
    This follows almost directly from the definition. The assumption $\theta \le \frac{\pi}{2}$ implies $x \ne 0$ and $\cos \theta \le p|_x / \|p\|, \forall p \in P$. Then $\cos \theta \|x\| \le p|_x$ is immediate since $\|x\| \le \|p\|$; see \Cref{fig:shortest-point}. Thus $x$ is a $(\cos\theta)$-approximation and also is a $(1 - \cos\theta)$-coreset of $P$.
\end{proof}

\subparagraph{Hardness of merging.}
We next show this simple $\eps = 1-\cos \theta$ bound cannot be significantly improved.  In particular, merging coresets with smaller error may obtain this error (Theorem \ref{thm:coreset}) and even merging $0$-error coresets may result in nearly this much error (Theorem \ref{thm:0coreset}). 

\begin{theorem}
    \label{thm:coreset}
    Consider a point set $P$ of angular diameter $\theta \le\frac{\pi}{2}$.  Decompose $P$ into $P_1$ and $P_2$.  There exists such a setting where 
    (1) $S_1$ is a $(1-\cos \frac{\theta}{2})$-coreset of $P_1$, 
    (2) $S_2$ is a $(1-\cos \frac{\theta}{2})$-coreset of $P_2$, 
    (3) $S$ is a $(1-\cos \frac{\theta}{2})$-coreset of $S_1 \cup S_2$, but 
    (4) $S$ is no better than a $(1-\cos \theta)$-coreset of $P$.
\end{theorem}
\begin{proof}
    We prove this existence by an example. Let $P$ include 3 points, $p_1$, $p_2$, and $p_3$. Such that $\|p_1\| = \|p_2\| = \|p_3\|$, $\angle(p_1, p_2) = \angle(p_2, p_3) = \theta/2$, and $\angle(p_1, p_3) = \theta$; see Figure \ref{fig:epsilon-coreset}. Then for $P_1 = \{p_2, p_3\}$, $S_1 = \{p_2\}$ is a valid $\left(1 - \cos\frac\theta 2\right)$-coreset.  For $P_2 = \{p_1\}$ then clearly $S_1 = \{p_1\}$ is a valid $\left(1 - \cos\frac\theta 2\right)$-coreset.  Now let $S = \{p_1\}$, so that $S$ is a valid $\left(1 - \cos\frac\theta 2\right)$-coreset of $S_1 \cup S_2 = \{p_2, p_1\}$.   However, $\frac{p_3|_{p_1}}{\|p_1\|} = \cos\theta$.  Therefore $S$ is not better than $\left(1 - \cos \theta \right)$-coreset of $P$.
\end{proof}

\begin{theorem}
    \label{thm:0coreset}
    Consider a point set $P$ of angular diameter $\theta \le\frac{\pi}{2}$.  Decompose $P$ into $P_1$ and $P_2$.  There exists a setting where 
    (1) $S_1$ is a $0$-coreset of $P_1$, 
    (2) $S_2$ is a $0$-coreset of $P_2$, 
    (3) $S$ is a $0$-coreset of $S_1 \cup S_2$, but 
    (4) $S$ is no better than a $\left( \frac{1 -\cos\theta}{1 + \cos\theta}\right)$-coreset of $P$.
\end{theorem}

\begin{proof}
    The proof is similar with the one of \cref{thm:coreset}. Let $P$ include 3 points, $p_1$, $p_2$, and $p_3$. Such that $p_2|_{p_1} = \|p_1\|$ and $p_3|_{p_2} = \|p_2\|$, also $\angle(p_1, p_2) = \angle(p_2, p_3) = \theta/2$, and $\angle(p_1, p_3) = \theta$; see \cref{fig:0coreset}. Then $S_1 = \{p_2\}$ is a 0-coreset of $P_1 = \{p_2, p_3\}$, and $S_2 = P_2 = \{p_1\}$ is a $0$-coreset for $P_2$.  Then $S = \{p_1\}$ is a $0$-coreset for $S_1 \cup S_2$.  However $\frac{p_3|_{p_1}}{\|p_1\|} = 1 - \frac{1 -\cos\theta}{1 + \cos\theta}$. Therefore $S$ can be no better than a $\left( \frac{1 -\cos\theta}{1 + \cos\theta}\right) $-coreset of $P$.
\end{proof}

\bibliography{references}
	
\end{document}